\documentclass[10pt,twocolumn,twoside]{IEEEtran}
\usepackage{textcomp}
\usepackage{amsmath,amssymb,amsfonts}
\usepackage{amsthm}
\usepackage{mathtools}
\usepackage{cite}
\usepackage{algorithm}
\usepackage{algpseudocode}
\usepackage{graphicx}
\usepackage{subcaption}
\usepackage{array,longtable}
\usepackage{textcomp}
\usepackage{siunitx}
\usepackage[dvipsnames]{xcolor}
\usepackage{pgfplots} 
\pgfplotsset{compat=newest}
\usepackage{tikz}
\usepackage[acronym]{glossaries} 
\usepackage[hidelinks, raiselinks=true, bookmarksnumbered=true, pdfusetitle, hyperfootnotes=false]{hyperref}
\usepackage[capitalise]{cleveref}
\crefname{equation}{}{} 
\def\BibTeX{{\rm B\kern-.05em{\sc i\kern-.025em b}\kern-.08em
    T\kern-.1667em\lower.7ex\hbox{E}\kern-.125emX}}

\newtheorem{theorem}{Theorem}
\newtheorem{remark}{Remark}

\newtheorem{lemma}{Lemma}

\newtheorem{proposition}{Proposition}
\newtheorem{definition}{Definition}

\newcommand{\onesVector}{\mathbf{1}}
\newcommand{\R}{\mathbb{R}}
\newacronym{CBF}{CBF}{control barrier function}
\newacronym{MPC}{MPC}{model predictive control}
\newacronym{PMSM}{PMSM}{permanent magnet synchronous motor}
\newacronym{QP}{QP}{quadratic program}

\begin{document}
\title{Tunable Real-Time Safety Filters via Set-Based Control Barrier Functions}
\author{Kim P. Wabersich, Felix Berkel, Felix Gruber, Sven Reimann
\thanks{The authors are with Corporate Research Robert Bosch GmbH, 71272 Renningen, Germany. (E-mail: {firstname.lastname}@bosch.com).}
}
\maketitle

\begin{abstract}
	Safety filters for industrial constrained systems are required to combine certified constraint satisfaction, predictable online computation, and a transparent tuning interface. Existing set-based filters are based on a well-established control invariant set design that scales favorably with state and input constraints, but typically intervene only at the set boundary. Control barrier function (CBF)-based filters, by contrast, provide tunable intervention but require a scalar barrier construction. This paper proposes a set-based CBF safety filter that turns a convex control invariant set directly into a tunable barrier via its Minkowski functional. The resulting filter is formulated as a single-level quadratic program (QP) in which one class-$\mathcal{K}^e$ parameter sets the intervention aggressiveness. Explicit convex formulations are derived for polytopic, zonotopic, and MPC-based invariant sets. Under standard bounded-disturbance assumptions, the resulting safety filter guarantees constraint satisfaction and asymptotic recovery into the invariant set. For tight real-time budgets, a learning-based approximation enables online acceleration, while the formal safety guarantees remain tied to the exact formulation. The method is validated in numerical studies and on a permanent-magnet synchronous motor drive, where an explicit QP implementation evaluates within a \SI{150}{\micro\second} sampling window and has a worst-case execution time of \SI{28.04}{\micro\second}.
\end{abstract}


\section{Introduction}
\label{sec:introduction}

Modern control systems are subject to tight state and input constraints, from collision avoidance in robotics to current limits in power electronics. Redesigning a well-tuned nominal controller whenever constraints change is impractical, which makes \emph{safety filters}~\cite{wabersich2023} a useful modular layer for embedded deployment. A safety filter monitors the nominal input and overrides it only when a constraint violation would otherwise occur.

The practical deployment of such filters is governed by three requirements.
First, the filter must be \emph{computationally tractable}, i.e., real-time feasible with bounded execution time within the sampling period.
Second, it must \emph{guarantee constraint satisfaction} and provide certified recovery after perturbations.
Third, it should offer \emph{tunable intervention}, with a transparent mechanism to trade off early corrections against late, aggressive overrides near the safe-set boundary.
Satisfying all three simultaneously is difficult. High-frequency applications such as electric drives with a \SI{150}{\micro\second} sampling period leave virtually no computational headroom, yet still demand certified constraint enforcement and smooth transients.

Existing approaches to safety filtering generally fall into three families, each based on set invariance~\cite{blanchini2015set}. \textit{Set-based safety filters}~\cite{fisac2018general,gruber2023scalable} verify state containment in a pre-computed invariant set and scale well to high-dimensional linear systems. They typically intervene abruptly at the set boundary and lack a recovery mechanism. \textit{\Glspl{CBF}-based safety filters}~\cite{ames2016control,alan2023control} offer smooth, tunable intervention by enforcing a lower bound on the one-step change of a scalar barrier function. The systematic synthesis of a scalar barrier that remains valid over a complex intersection of constraints, however, scales poorly or introduces conservatism. \textit{Predictive safety filters}~\cite{wabersich2021probabilistic} solve a receding-horizon problem at every step, offering flexibility at the cost of substantial online computation. \cref{tab:filter_comparison} summarizes these paradigms.

\par

In this work, we close the gap between scalable invariant-set computations and the online tuning interface of \glspl{CBF}. Our starting point is the observation that the Minkowski functional of any convex control invariant set is itself a valid \gls{CBF}. Using this fact, we formulate the safety filter as a convex optimization problem in which a single class-$\mathcal{K}^e$ parameter governs the aggressiveness of the intervention. The offline design inherits the scalability of established set-computation tools (e.g., zonotopic methods~\cite{gruber2023scalable,schafer2024}, reachability analysis~\cite{gruber2023scalable}, or predictive terminal sets~\cite{rawlings2017model,wabersich2018linear}). Online, the filter reduces to a convex \gls{QP} that endows any such set with tunable intervention and asymptotic recovery (\cref{tab:filter_comparison}).

\par

\begin{table}[t]
\centering
\footnotesize
\setlength{\tabcolsep}{2.5pt}
\renewcommand{\arraystretch}{1.15}
\caption{Feature comparison of safety-filter families for linear systems with convex constraints (per~\cite{wabersich2023}).}
\label{tab:filter_comparison}
\begin{tabular}{lccccc}
\hline
\textbf{Family} & \textbf{Design} & \textbf{Scalable} & \textbf{Tunable} & \textbf{Recovery} & \textbf{Online} \\
\hline
Set-based\textsuperscript{a} & explicit inv.\ set & $\checkmark$ & $\times$ & $\times$ & set check \\
CBF\textsuperscript{b} & barrier & $\times$\textsuperscript{\dag} & $\checkmark$ & $\checkmark$ & QP \\
Predictive\textsuperscript{c} & model + term.\ set & $\checkmark$ & indirect & $\times$ & MPC problem\\
\textbf{Set-based CBF} & inv.\ set\textsuperscript{\ddag} & $\checkmark$ & $\checkmark$ & $\checkmark$ & QP \\
\hline
\end{tabular}
\par\vspace{2pt}\raggedright
{\scriptsize \textsuperscript{a}\cite{fisac2018general,gruber2023scalable}\enspace\textsuperscript{b}\cite{ames2016control,alan2023control}\enspace\textsuperscript{c}\cite{rawlings2017model,wabersich2021probabilistic,wabersich2023}\\
\textsuperscript{\dag}Existing synthesis (SOS, HJ) scales poorly with many coupled constraints.\\
{\textsuperscript{\ddag}Explicit or implicit (e.g., \gls{MPC} feasible set without enumeration).}}
\end{table}

\begin{figure}[t]
\centering
\includegraphics[width=0.75\linewidth]{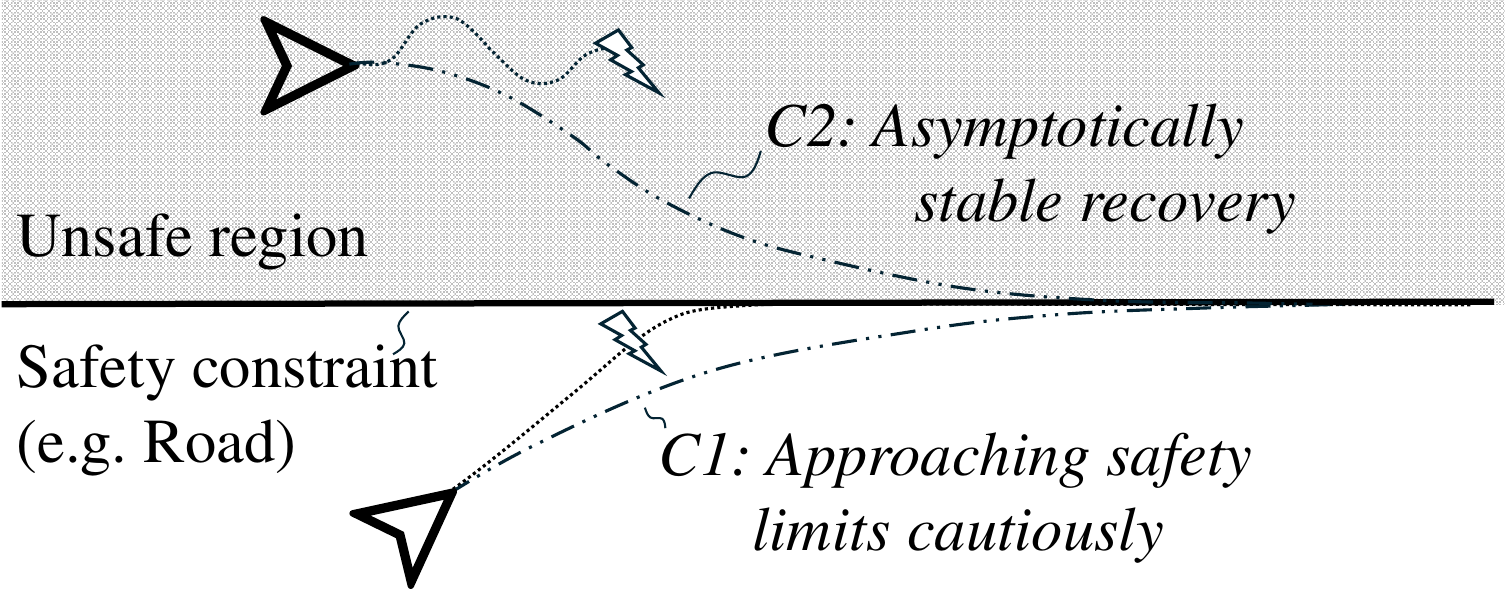}
\caption{Illustration of safety filter behaviors near the boundary of safe states (shaded: unsafe region). Dotted lines show set-based safety filters with aggressive interventions and no recovery guarantees.
\textbf{Contribution 1:} Our set-based \gls{CBF} safety filter (solid line) enables tunable, dampened approach to the boundary.
\textbf{Contribution 2:} If the boundary is crossed, our method (dashed-dotted line) guarantees asymptotic recovery, a guarantee not provided by set-based safety filters.}
\label{fig:safety_filter_comparison}
\end{figure}

\paragraph*{Contributions}
We propose a safety-filter design framework for linear systems with convex constraints, with the following contributions (\cref{fig:safety_filter_comparison}).

\par
\textit{C1) \Gls{CBF} construction from invariant sets}: We show that for any convex control invariant set (e.g., polytope, zonotope, or \gls{MPC} feasible set), its Minkowski functional is a valid \gls{CBF}. This construction preserves the exact constraint geometry without conservatism, and a single class-$\mathcal{K}^e$ parameter sets the boundary-approach dynamics (\cref{prop:nominal_set_cbf}).

\par
\textit{C2) Guaranteed recovery}: If the invariant set satisfies a robust invariance condition, the proposed filter drives out-of-bound states asymptotically back into the safe region, endowing the safe set with a domain of attraction (\cref{thm:set_based_cbf}).

\par
\textit{C3) Real-time convex implementation}: We derive exact convex reformulations of the safety-filter \gls{QP} for standard set representations (\cref{sec:design_and_implementation}). For polytopic sets, the parametric \gls{QP} admits an explicit offline solution. For more complex representations, we provide a learning-based approximation that reduces the online footprint to a standard \gls{CBF} safety-filter evaluation with bounded safety errors (\cref{subsec:approximation_set_based_cbf}).

\par
The approach is validated in simulation on a 10-state pendulum chain and a motion-control task (\cref{sec:numerical_simulation}), and experimentally on an electric drive operating at a \SI{150}{\micro\second} sampling rate (\cref{sec:electrical_machine}).

\subsection{Related Work}
We survey the relevant literature on synthesizing minimally invasive, real-time safety filters across four main categories as follows.

\paragraph*{Invariant-set computation}
Reachability analysis~\cite{gruber2023scalable}, zonotopic methods~\cite{schafer2024}, and system-level synthesis~\cite{leeman2023predictive} compute control invariant sets for constrained linear systems. These methods yield certified safe sets offline, but lack a tuning mechanism for online boundary interventions. They provide the geometric objects for our approach; we augment them with a scalar barrier derived from the set geometry to enable tunable intervention online.

\paragraph*{\Gls{CBF}-based safety filters}
\Gls{CBF} methods~\cite{ames2016control,alan2023control} impose a lower bound on the one-step change of a barrier function along closed-loop trajectories. While broadly applicable to nonlinear systems, systematic synthesis tools such as SOS programming or HJ reachability scale poorly with coupled constraints or introduce conservatism through smooth safe-set approximations. Our approach avoids this trade-off: by constructing the barrier directly from an available control invariant set, we inherit its exact constraint geometry while retaining the tuning interface of \glspl{CBF}.

\paragraph*{Predictive safety filters and robust \gls{MPC}}
Predictive safety filters~\cite{wabersich2021probabilistic,wabersich2023} and robust \gls{MPC}~\cite{rawlings2017model} solve finite-horizon optimizations to recursively enforce feasibility. They require a prediction model, tightened constraints, and a terminal set. Tuning the intervention aggressiveness is indirect and is handled via cost weights and horizon length. Our formulation reuses the same feasible-set object but introduces a barrier-step constraint, yielding a single, explicit tuning parameter.

\paragraph*{Hamilton-Jacobi barrier-value functions}
The work in~\cite{choi2021robust} constructs barrier-value functions via Hamilton-Jacobi reachability, providing strong safety guarantees for nonlinear systems. Solving the underlying HJ partial differential equation suffers from the curse of dimensionality, typically limiting tractability to systems below six states. Our restriction to linear systems with convex constraints avoids this bottleneck and enables the use of efficient, implicit set representations.

\paragraph*{Outline of this article}
After formalizing the system class, constraints, and safety-filter problem (\cref{sec:problem_setting}), we recall discrete-time \glspl{CBF} and their use in safety filters (\cref{sec:cbf}). Building on this, \cref{sec:set_based_cbf} embeds the \gls{CBF} mechanism into set-based filters, and \cref{sec:design_and_implementation} details how to construct these filters from efficiently computed control invariant sets. To accommodate tight real-time budgets, \cref{subsec:approximation_set_based_cbf} introduces a learned approximation with input-to-state stability guarantees. The framework is validated via simulations (\cref{sec:numerical_simulation}) and hardware experiments (\cref{sec:electrical_machine}).

\paragraph*{Notation}
For a set $\mathcal A\subseteq\R^n$ and scalar $c\in\R$, the scaled set is denoted by $c\mathcal A\triangleq \{ca | a\in\mathcal A\}$. A continuous function $\alpha: [0, a) \rightarrow \R$ for some $a>0$ is class $\mathcal K$ $(\alpha \in \mathcal K)$ if it is strictly increasing and $\alpha(0) = 0$, and is extended class $\mathcal K$ ($\alpha\in\mathcal K^e$) if it is a class $\mathcal K$ function defined on $(-a,b)$ with $a,b>0$. A ball of radius $r>0$ defined by the norm $\Vert . \Vert$ is denoted by $\mathcal B_n(r)\triangleq \{ x\in\mathbb R^n~|~\Vert x \Vert \leq r\}$. The Minkowski sum of two sets $\mathcal A, \mathcal B\subseteq\R^n$ is defined as $\mathcal A\oplus \mathcal B\triangleq \{a+b~|~a\in\mathcal A, b\in\mathcal B\}$.
The Pontryagin difference of two sets $\mathcal A, \mathcal B\subseteq\R^n$ is defined as $\mathcal A\ominus \mathcal B\triangleq \{x\in\R^n~|~\{x\}\oplus\mathcal B\subseteq\mathcal A\}$. The interior of a set $\mathcal A \subseteq \R^n$ is denoted by $\mathrm{int}(\mathcal A)$.
We adopt the following shorthand: $\mathcal{D}$ denotes the \gls{CBF} domain, $\mathcal{S}\triangleq\{x\in\mathcal D:h(x)\geq 0\}$ is the safe set, $\Omega$ is a control invariant set, and $\tilde{\Omega}$ is a robust control invariant set. Since $\mathcal{S}$ always equals $\Omega$ and $\mathcal{D}$ is either $\Omega$ or $\tilde{\Omega}$, both $\mathcal{S}$ and $\mathcal{D}$ are dropped in favour of $\Omega$ and $\tilde{\Omega}$ after \cref{sec:cbf}.

\section{Preliminaries}\label{sec:preliminaries}

\subsection{Problem Setting}\label{sec:problem_setting}

We consider discrete-time, linear time-invariant (LTI) systems of the form
\begin{align}\label{eq:linear_system}
x(k+1) = Ax(k) + Bu(k),
\end{align}
where $x(k)\in\mathbb R^{n_x}$ is the system state and $u(k)\in \mathbb R^{n_u}$ is the control input at time step $k\in\mathbb N$.
The system is subject to hard state and input constraints
\begin{align}\label{eq:constraints}
x(k) \in \mathcal X, \qquad
u(k) \in \mathcal U,
\end{align}
where $\mathcal X \subset \mathbb R^{n_x}$ and $\mathcal U \subset \mathbb R^{n_u}$ are compact, convex sets containing the origin in their respective interiors.

\par

A safety filter acts as a minimally invasive supervisory map $\kappa_{\mathrm f}:\R^{n_x}\times\R^{n_u}\rightarrow\mathcal U$. It monitors the current state $x(k)$ and a proposed nominal input $u_{\mathrm{des}}(k)$, passes $u_{\mathrm{des}}(k)$ to the plant unmodified whenever constraint satisfaction can be certified, and overrides it with the minimal necessary deviation otherwise.
The synthesis of a practical safety filter requires balancing theoretical guarantees with deployment constraints. This motivates the following formal design problem.

\par

Given the LTI system~\eqref{eq:linear_system}, convex constraints~\eqref{eq:constraints}, a nominal input sequence $u_{\mathrm{des}}(k)$, and a known control invariant set $\Omega\subseteq\mathcal X$ with $0\in\mathrm{int}(\Omega)$, design a safety filter $u(k) = \kappa_{\mathrm f}(x(k),u_{\mathrm{des}}(k))$ that satisfies the following properties:
\begin{enumerate}
\item[(P1)] \textit{Constraint satisfaction:} For any initial condition $x(0)\in\Omega$, the closed-loop system guarantees $x(k)\in\mathcal X$ and $u(k)\in\mathcal U$ for all $k\in\mathbb N$.
\item[(P2)] \textit{Tunable intervention:} The filter provides an explicit parameter to shape the boundary-approach dynamics, dictating the trade-off between nominal performance and early safety interventions.
\item[(P3)] \textit{Recovery from unsafe states:} If equipped with a robust control invariant set $\tilde\Omega\supset\Omega$ (\cref{sec:robust_invariance_and_stability}), the filter asymptotically steers trajectories initialized in $\tilde\Omega\setminus\Omega$ back into $\Omega$.
\item[(P4)] \textit{Real-time implementation:} The filter reduces to a convex optimization problem compatible with standard set representations, supporting strict real-time requirements.
\end{enumerate}

Contributions C1-C3 address these requirements sequentially: C1 establishes (P1) and (P2) via \cref{prop:nominal_set_cbf}; C2 achieves the recovery requirement (P3) via \cref{thm:set_based_cbf}; and C3 provides the convex and learning-based architectures requisite for (P4) in \cref{sec:design_and_implementation,subsec:approximation_set_based_cbf}. To lay the groundwork, we first formalize discrete-time \glspl{CBF} and the foundational notions of set invariance.

\subsection{Safety Filters based on Control Barrier Functions}\label{sec:cbf}

\Gls{CBF}-based safety filters synthesize a safe control action by solving a minimally invasive projection problem~\cite{gurriet2018towards,choi2021robust,didier2023approximate}:
\begin{subequations}
\label{eq:cbf_filter}
\begin{align}
\min_{u\in\mathcal U} &~G(u, u_{\mathrm{des}}) \label{eq:cbf_filter_objective} \\
\mathrm{s.t.} \; & ~h(Ax + Bu) - h(x) \geq \Delta h(x),\label{eq:cbf_filter_step_condition}
\end{align}
\end{subequations}
yielding the filtering law
\begin{align}\label{eq:cbf_safety_filter_law}
\kappa_{\mathrm{f}}(x,u_{\mathrm{des}})\triangleq u^*(x,u_{\mathrm{des}}),
\end{align}
where $u^*(x,u_{\mathrm{des}})$ is an optimizer of~\eqref{eq:cbf_filter}.

The objective $G(u, u_{\mathrm{des}})$ typically minimizes a specified norm, e.g., $G(u, u_{\mathrm{des}}) = \Vert u - u_{\mathrm{des}} \Vert$, although more complex metrics that anticipate future desired inputs are also possible~\cite{wabersich2023}. Safety is enforced by the discrete-time \gls{CBF} constraint~\eqref{eq:cbf_filter_step_condition}. Following~\cite[Definition III.1]{wabersich22} and~\cite{agrawal2017discrete}, the scalar function $h(x)$ encodes a continuous safety margin: positive values certify safety, while negative values indicate constraint violations. The step bound $\Delta h(x)$ shapes the dynamic evolution near the boundary.

\begin{definition}[Control Barrier Function]\label{def:barrier_function}
A continuous function $h:\mathcal D\rightarrow\R$ defined on a compact domain $\mathcal D\subset\R^{n_x}$ is a \gls{CBF} for~\eqref{eq:linear_system} under constraints~\eqref{eq:constraints} if its corresponding safe set $\mathcal S\triangleq \{x\in\mathcal D:h(x)\geq 0\}$ satisfies $\mathcal S\subseteq\mathcal X$, both $\mathcal S$ and $\mathcal D$ are non-empty and compact, and there exists a function $\alpha\in\mathcal K^e$ such that for all $x\in\mathcal D$ there exists an input $u\in\mathcal U$ with $Ax+Bu\in\mathcal D$ satisfying
\begin{align}\label{eq:cbf_condition}
h(Ax + Bu)-h(x) \geq \Delta h(x),
\end{align}
with the step bound defined as
\begin{align}\label{eq:cbf_step_bound}
\Delta h(x) \triangleq
\begin{cases}
-\min(\alpha(h(x)),h(x)), & \text{if } h(x)\geq 0, \\
-\max(\alpha(h(x)),h(x)), & \text{else.}
\end{cases}
\end{align}
The set of safe control inputs at $x\in\mathcal D$ is subsequently defined as
\begin{align}\label{eq:cbf_safe_control}
K_{\mathrm{CBF}}(x)\triangleq \{u\in\mathcal U\;|\;\cref{eq:cbf_condition}\}.
\end{align}
\end{definition}

\par

The two-case $\min/\max$ formulation of $\Delta h(x)$ in \cref{def:barrier_function} serves two roles. For states interior to the safe set ($x\in\mathcal S$), the $\min$ operator ensures $h(Ax+Bu) \geq 0$, establishing forward invariance while allowing bounded margin decay. For out-of-bound states ($x\notin\mathcal S$), the $\max$ operator ensures that the required one-step increase in $h$ is governed by $\alpha$ rather than $h(x)$ itself, so the filter demands monotone progress toward $\mathcal S$ without requiring an immediate jump to its boundary (\cref{fig:safety_filter_comparison}).

\par

Compared to~\cite[Definition III.1]{wabersich22}, we enforce the descent condition~\eqref{eq:cbf_condition} globally over $\mathcal D$ (rather than only on $\mathcal D\setminus\mathcal S$) using the explicit structure of~\eqref{eq:cbf_step_bound}. This generalizes discrete-time exponential barriers~\cite{agrawal2017discrete}, in which $\alpha(h(x))=\gamma h(x)$, while inheriting the class-$\mathcal K^e$ tuning flexibility of continuous-time \gls{CBF} design~\cite{ames2016control,alan2023control}. The function $\alpha$ acts as a direct tuning parameter: a linear choice $\alpha(r) = s\cdot r$ bounds the margin decay geometrically, where $s \in (0, 1]$ shifts the filter behavior from a permissive invariance check ($s=1$) to a damped, early-intervention supervisor ($s \to 0$).

\begin{definition}[Set Invariance]\label{def:control_invariant_set}
(i) A set $\Omega\subseteq\mathbb R^{n_x}$ is \textit{control invariant} for system~\eqref{eq:linear_system} under~\eqref{eq:constraints} if, for all $x\in\Omega$, there exists $u\in\mathcal U$ such that $Ax + Bu\in\Omega$. 
(ii) For an autonomous system operating under a fixed policy $\pi:\mathbb R^{n_x}\rightarrow\mathbb R^{n_u}$, the set is \textit{invariant} if $Ax + B\pi(x)\in\Omega$ for all $x\in\Omega$.
\end{definition}

\begin{proposition}\label{prop:cbf_properties}
Let $\mathcal D\subset\mathbb R^{n_x}$ be a non-empty, compact set, and let $h:\mathcal D\rightarrow\mathbb R$ be a \gls{CBF} with safe set $\mathcal S$ per \cref{def:barrier_function}. If $\mathcal S \subset \mathcal D$ and $\mathcal D$ is invariant for~\eqref{eq:linear_system} under a policy $u(k)=\kappa(x(k)) \in K_{\mathrm{CBF}}(x)$, then:
\begin{enumerate}
\item $\mathcal S$ is a forward invariant set,
\item $\mathcal S$ is asymptotically stable within $\mathcal D$.
\end{enumerate}
\end{proposition}
\begin{proof}
Because the condition~\eqref{eq:cbf_condition} imposes a strictly tighter bound than the original definition in~\cite[Definition III.1]{wabersich22}, the proof follows directly from~\cite[Theorem III.4]{wabersich22}.
\end{proof}

The properties summarized in \cref{prop:cbf_properties} capture the core utility of \gls{CBF}-based safety filters. Synthesizing such barrier functions analytically is generally difficult; the next sections show how they can be constructed systematically from established control invariant sets.

\section{Set-Based Control Barrier Functions}
\label{sec:set_based_cbf}

This section presents a mechanism to construct a \gls{CBF} directly from a control invariant set $\Omega$. The set may be represented explicitly (e.g., by half-spaces) or implicitly (e.g., as the feasible domain of a robust predictive controller). In the implicit case, the set need not be projected or enumerated; the barrier value is computed via scaled set membership, embedding the invariant set's geometry directly into the QP constraints.

This decoupling separates two design tasks: (i) obtaining a control invariant set with established, scalable methods (\cref{sec:design_and_implementation}), and (ii) converting that static set into a tunable \gls{CBF} safety filter. The underlying set-theoretic tools (Minkowski functionals, C-set scaling, contractive Lyapunov arguments) are classical~\cite{blanchini2015set}; our contribution is their combination into an online, class-$\mathcal K$-tunable \gls{CBF} safety filter, including the robust-superset recovery mechanism (\cref{sec:robust_invariance_and_stability}) and its convex implementation (\cref{sec:design_and_implementation}).

\par

Let $\Omega \subseteq \mathcal{X}$ be a compact, convex, and control invariant set for system~\eqref{eq:linear_system} under constraints~\eqref{eq:constraints}, containing the origin in its interior ($0\in\mathrm{int}(\Omega)$). For any state $x \in \R^{n_x}$, the Minkowski functional~\cite[Section 3.1.2]{blanchini2015set} is defined as
\begin{align}
\gamma_{\Omega}(x)\triangleq \inf\{\gamma \geq 0:~x\in\gamma\Omega\}. \label{eq:set_based_gamma}
\end{align}
We define the corresponding set-based \gls{CBF} as
\begin{align}
h(x) \triangleq 1 - \gamma_{\Omega}(x). \label{eq:set_based_cbf}
\end{align}
For this set-based construction, the generic safe set in \cref{def:barrier_function} exactly recovers the invariant set: $\{x:h(x)\geq 0\} = \Omega$. In the nominal case (\cref{prop:nominal_set_cbf}), the \gls{CBF} is defined on the domain $\mathcal D = \Omega$. When out-of-set recovery is required (\cref{thm:set_based_cbf}), the domain expands to a robust control-invariant superset, $\mathcal D = \tilde\Omega\supset\Omega$. Hence, the remainder of the construction uses $\Omega$ and $\tilde\Omega$ directly rather than the generic symbols $\mathcal S$ and $\mathcal D$.

\begin{proposition}\label{prop:nominal_set_cbf}
Let $\Omega$ be a convex control invariant set satisfying $0\in\mathrm{int}(\Omega) $ and $\Omega\subseteq \mathcal X$. The corresponding set-based \gls{CBF}~\eqref{eq:set_based_cbf} satisfies \cref{def:barrier_function} on domain~$\Omega$, establishing forward invariance of~$\Omega$.
\end{proposition}

The proof, provided in the appendix, shows that deploying the set-based \gls{CBF}~\eqref{eq:set_based_cbf} in the safety filter~\eqref{eq:cbf_filter} guarantees forward invariance for all $x(0)\in\Omega$.

The operational flexibility of this construction comes from the arbitrary choice of $\alpha\in\mathcal K^e$ used to shape the step bound $\Delta h$, which lets the designer specify the closed-loop approach to the safe-set boundary (\cref{fig:safety_filter_comparison}). We demonstrate this design freedom in \cref{subsec:motion_control,sec:electrical_machine}.

To make the intuition precise, consider a linear choice $\alpha(r) = s\cdot r$ with $s\in(0,1]$. The filter constraint~\eqref{eq:cbf_filter_step_condition} enforces
\begin{align}\label{eq:tuning_bound}
h(x(k+1)) \;\geq\; (1-s)\,h(x(k)),
\end{align}
which bounds the per-step geometric decay of the safety margin. At the extremes, $s=1$ relaxes the constraint to $h(x(k{+}1))\geq 0$ and recovers a standard set-invariance filter that only prevents boundary violations. Conversely, $s\to 0$ enforces $h(x(k{+}1))\approx h(x(k))$ and yields conservative, highly damped trajectories. The parameter $s$ thus provides a single design knob: values close to $1$ prioritize the nominal controller's tracking fidelity, whereas smaller values produce early, smooth intervention. The latter is useful for systems with limited actuator bandwidth, significant measurement noise, or strict safety-buffer requirements.

\begin{remark}[Lyapunov perspective]\label{rem:lyapunov}
The set-based \gls{CBF} $h(x) = 1 - \gamma_\Omega(x)$ acts both as a safety margin and as a Lyapunov-like function for the safe set $\Omega$. For $x\notin\Omega$ (i.e., $h(x)<0$), the bound~\eqref{eq:cbf_step_bound} requires $h$ to strictly increase per step, which is equivalent to a monotone contraction of $\gamma_\Omega(x)$ toward $1$. For $x\in\Omega$ (i.e., $h(x)\geq 0$), the bound transitions to a decay-rate constraint on the margin. We exploit this dual role in \cref{thm:set_based_cbf} to establish asymptotic stability.
\end{remark}

\subsection{Robust Invariance for Asymptotic Stability of the Safe Set}\label{sec:robust_invariance_and_stability}

We extend the nominal framework of \cref{prop:nominal_set_cbf} to guarantee asymptotic stability of the safe set $\Omega$, i.e., the second case of \cref{eq:cbf_step_bound}. This requires a \textit{robust} control invariant set $\tilde \Omega$ computed for an artificial additive disturbance set $\mathcal W$.

\begin{definition}[Robust Control Invariant Set]\label{def:robust_control_invariant_set}
(i) A set $\tilde\Omega\subseteq \R^{n_x}$ is \textit{robust control invariant} for system~\eqref{eq:linear_system} under~\eqref{eq:constraints} with respect to an additive disturbance set $\mathcal W\subseteq\R^{n_x}$ if, for all $x\in\tilde\Omega$, there exists $u\in\mathcal U$ such that $Ax + Bu + w\in\tilde\Omega$ for all $w\in\mathcal W$.
(ii) For autonomous systems under a fixed feedback controller $\pi:\mathbb R^{n_x}\rightarrow\mathbb R^{n_u}$, the set is \textit{robust invariant} if $Ax+B\pi(x)+w\in\tilde \Omega$ for all $x\in\tilde\Omega$ and $w\in\mathcal W$.
\end{definition}

The robustness margin encoded in $\tilde\Omega$ provides the control authority needed to asymptotically stabilize a strictly contracted invariant set $\Omega \triangleq \nu \tilde\Omega$ with $\nu\in(0,1)$. Choosing $\nu$ such that $\tilde\Omega \ominus \mathcal W\subseteq \Omega$ gives the filter the capacity to steer states from the neighborhood $\tilde\Omega\setminus\Omega$ back into $\Omega$. Deviations outside $\Omega$ are treated as virtual disturbances that the robust set is designed to reject.

\begin{theorem}\label{thm:set_based_cbf}
Let $\tilde\Omega\subset \R^{n_x}$ be a convex robust control invariant set with $0\in\mathrm{int}(\tilde \Omega)$ and $\tilde \Omega \subseteq \mathcal X$ for some disturbance $\mathcal W$ with $0\in\mathrm{int}(\mathcal W)$.
Let $\Omega \triangleq \nu \tilde\Omega$ for some $\nu \in (0,1)$ such that $\tilde\Omega \ominus \mathcal W\subseteq \Omega$ holds.
Then, the set-based \gls{CBF}~\eqref{eq:set_based_cbf} satisfies \cref{def:barrier_function} on domain~$\tilde\Omega$, establishing forward invariance of~$\Omega$ and asymptotic stability of~$\Omega$ in~$\tilde\Omega$.
\end{theorem}
\begin{proof}
The proof builds upon \cref{prop:nominal_set_cbf}. Control invariance of $\Omega=\nu\tilde\Omega$ follows from robust invariance of $\tilde\Omega$ at $w=0$ combined with convexity of $\mathcal U$ and $0\in\mathcal U$, so \cref{prop:nominal_set_cbf} applies and establishes forward invariance of $\Omega$.
It remains only to verify the second case of the bound \cref{eq:cbf_step_bound}.
When $h(x)<0$, \cref{eq:cbf_step_bound} yields $\Delta h(x) = -\max(\alpha(h(x)), h(x)) \leq -h(x)$. It thus suffices to establish $h(Ax+Bu) - h(x) \geq -h(x)$, which simplifies to $h(Ax + Bu ) \geq 0$. Since $\tilde\Omega$ is robustly control invariant, there exists an input $u\in\mathcal U$ such that $Ax+Bu+w\in\tilde\Omega$ for all $w\in\mathcal W$. By definition of the Pontryagin difference, $Ax+Bu\in\tilde\Omega\ominus \mathcal W \subseteq \Omega$ (by assumption). Because $\Omega$ is convex and $0\in\mathrm{int}(\Omega)$, membership $Ax+Bu\in\Omega$ implies $\gamma_{\Omega}(Ax+Bu)\leq 1$~\cite[Section~3.1.2]{blanchini2015set}, ensuring $h(Ax+Bu)=1-\gamma_\Omega(Ax+Bu) \geq 0$.
\end{proof}

The choice of $\mathcal W$ and $\nu$ parameterizes a trade-off between recovery capability and conservatism in nominal operation. A larger virtual disturbance $\mathcal W$ permits a smaller $\nu$, which expands the guaranteed domain of attraction $\tilde \Omega \setminus \nu\tilde\Omega$ (the recovery neighborhood). This shrinks the certified safe set $\Omega = \nu\tilde\Omega$ and forces the filter to act more restrictively during nominal operation. If out-of-bound recovery is not required, the nominal formulation of \cref{prop:nominal_set_cbf} is sufficient.

Both \cref{prop:nominal_set_cbf} and \cref{thm:set_based_cbf} require $0\in\mathrm{int}(\Omega)$, a standard assumption for controllers regulating to an interior equilibrium. This property is easy to verify computationally, e.g., by checking $\mathcal{B}_{n_x}(r)\subseteq\Omega$ for some $r>0$. Scalable set-computation algorithms, such as viability theory~\cite{liniger2017real}, zonotopic reachability~\cite{gruber2023scalable,schafer2024}, and robust predictive control~\cite{rawlings2017model}, produce sets that satisfy this condition and serve as plug-and-play inputs to the workflow in \cref{alg:workflow}.
\section{Implementation and Design of Set-based \texorpdfstring{\gls{CBF}}{CBF} Safety Filters}
\label{sec:design_and_implementation}

While the set-based \gls{CBF} introduced in \cref{sec:set_based_cbf} is representation-agnostic, deploying it as a real-time safety filter requires reformulating the Minkowski functional and set-scaling operations into tractable convex constraints. Scalable methods for computing control invariant sets provide these representations: classical invariance algorithms~\cite{blanchini2015set,herceg2013multi} yield H-polytopes, data-driven methods~\cite{rosolia2021robust,wabersich2021soft} yield V-polytopes, zonotopic reachability~\cite{gruber2023scalable,schafer2024} yields tight under-approximations for high-dimensional systems, and robust predictive control~\cite{rawlings2017model,wabersich2023} yields implicit feasible domains. This section derives explicit convex reformulations for each case.

Directly enforcing the \gls{CBF} step condition~\cref{eq:cbf_filter_step_condition} via~\cref{eq:set_based_cbf} requires evaluating $h(Ax+Bu)$, which is generally a bi-level optimization problem. To obtain a single-level program, we introduce an auxiliary decision variable $\gamma^+ \geq 0$ that bounds the Minkowski functional at the successor state:
\begin{subequations}
	\label{eq:set_based_cbf_filter}
	\begin{align}
		\min_{u\in\mathcal U, \gamma^+\geq 0} \quad & G(u, u_{\mathrm{des}}) \label{eq:set_based_cbf_filter_objective}                    \\
		\mathrm{s.t.} \quad                         & 1-\gamma^+ - h(x) \geq \Delta h(x),\label{eq:set_based_cbf_filter_step_condition} \\
		                                        & Ax + Bu \in \gamma^+ \Omega. \label{eq:set_based_cbf_filter_set_constraint}
	\end{align}
\end{subequations}
Constraint~\cref{eq:set_based_cbf_filter_set_constraint} guarantees that $1-\gamma^+$ is a valid lower bound for the successor \gls{CBF} value, i.e., $1-\gamma^+ \leq h(Ax+Bu)$. The reformulation is exact in the input variable: for any candidate $u$, the minimal feasible value of $\gamma^+$ is $\gamma_\Omega(Ax+Bu)$, which recovers the original \gls{CBF} constraint. Evaluating $h(x) = 1 - \gamma_\Omega(x)$ for the current state may require a separate optimization, but this cost is incurred only once per sampling instant and is decoupled from the prediction dynamics. As noted in~\cref{rem:reformulation_gamma_plus}, this evaluation can be avoided through objective regularization.

The remaining synthesis task is to formulate the geometric scaling constraint~\cref{eq:set_based_cbf_filter_set_constraint} for a given representation of $\Omega$. \Crefrange{subsec:poly_h_rep}{subsec:zonotopic_set_based_cbf} address explicit representations, in which scaled set membership maps directly to linear constraints. \Cref{subsec:predictive_safe_sets} treats the implicit MPC case and enforces membership via scalable prediction horizons. \Cref{subsec:approximation_set_based_cbf} introduces learning-based approximations that remove online optimization entirely. The overall offline-online procedure is summarized in \cref{alg:workflow}.

\begin{algorithm}[t]
\caption{Set-Based CBF Safety Filter Design}
\label{alg:workflow}
\begin{algorithmic}[1]
\Statex \textbf{Offline design (once):}
\State Obtain a control invariant set $\Omega$: classical invariance~\cite{blanchini2015set,herceg2013multi} (H-polytope), data-driven methods~\cite{rosolia2021robust,wabersich2021soft} (V-polytope), zonotopic reachability~\cite{gruber2023scalable,schafer2024}, or MPC design~\cite{rawlings2017model,wabersich2023}.
\State Verify $0\in\mathrm{int}(\Omega)$ (e.g., check $\mathcal{B}_{n_x}(r)\subseteq\Omega$ for an $r>0$).
\State Choose the class-$\mathcal{K}^e$ function $\alpha$ (e.g., $\alpha(r)=s\cdot r$, $s\in(0,1]$) per tuning guidance in \cref{sec:set_based_cbf}.
\State Formulate the safety-filter QP via the convex reformulation for the chosen representation (\cref{subsec:poly_h_rep,subsec:zonotopic_set_based_cbf,subsec:predictive_safe_sets}).
\State \textit{[Optional]} Solve parametric QP explicitly offline (\cref{subsec:poly_h_rep}) or train an approximation $\bar{h}$ (\cref{subsec:approximation_set_based_cbf}).
\Statex \hrulefill
\Statex \textbf{Online (each sampling step):}
\State Receive current state $x(k)$ and desired input $u_{\mathrm{des}}(k)$.
\State Evaluate $h(x(k))=1-\gamma_\Omega(x(k))$ and $\Delta h(x(k))$ via \cref{eq:cbf_step_bound}.
\State Solve (or look up) safety-filter QP~\cref{eq:set_based_cbf_filter}; obtain $u(k)=\kappa_{\mathrm{f}}(x(k),u_{\mathrm{des}}(k))$.
\State Apply $u(k)$ to the plant.
\end{algorithmic}
\end{algorithm}

\begin{remark}\label{rem:reformulation_gamma_plus}
	Adding a small penalty $\rho {(\gamma^+)}^2$, $\rho>0$, to the objective~\cref{eq:set_based_cbf_filter_objective} encourages that the optimal solution satisfies $\gamma^{+,*} = \gamma_\Omega(Ax+Bu)$. Consequently, $h(x)$ at the next time step can be approximated directly by $1 - \gamma^{+,*}$, removing the need for an explicit margin evaluation.
\end{remark}

\subsection{Polytopes in H-Representation}\label{subsec:poly_h_rep}
Standard algorithms for computing (robust) control invariant sets for linear systems with polytopic constraints yield sets in half-space representation (H-representation)~\cite{herceg2013multi}. A polytope containing the origin takes the form $\Omega = \{x\in\R^{n_x}\mid Hx\leq \onesVector\}$ with $H\in\R^{N_h\times n_x}$. Since $\gamma \Omega = \{x \in\R^{n_x}\mid Hx \leq \gamma \onesVector\}$, the set-scaling constraint~\cref{eq:set_based_cbf_filter_set_constraint} reduces to the linear inequality
\begin{equation}\label{eq:cbf_polytope_constraint}
	H(Ax + Bu) \leq \gamma^+ \onesVector.
\end{equation}

\subsection{Data-Driven Invariant Sets: Polytopes in V-Representation}
Data-driven robust control invariant sets, e.g., those synthesized for repetitive tasks~\cite{rosolia2021robust} or terminal trajectory bounding~\cite{wabersich2021soft}, are represented as polytopes in vertex form (V-representation): $\Omega = \left\{\sum_{i=1}^{N_v} \lambda _i v_i \;\middle|\; \sum_{i=1}^{N_v} \lambda_i = 1, \lambda_i\geq 0 \right\}$. If $0\in\Omega$, we append $v_0=0$ as a trivial vertex without changing the geometry.

Following~\cite[Lemma 2]{wabersich2021soft}, the Minkowski functional evaluates to $\gamma_{\Omega}(x) = \min_{\{\lambda_i\}} (1-\lambda_0)$ subject to the standard convex-combination constraints. This yields an exact implementation of~\cref{eq:set_based_cbf_filter_set_constraint} with non-negative multipliers $\lambda_i\geq 0$:
\begin{subequations}
	\begin{align*}
		Ax + Bu  = \sum_{i=0}^{N_v} \lambda_i v_i,\quad
		\sum_{i=0}^{N_v} \lambda_i = 1,\quad
		1-\lambda_0  = \gamma^+.
	\end{align*}
\end{subequations}

\subsection{Scalable Invariant Set Computations: Zonotopes}\label{subsec:zonotopic_set_based_cbf}
Zonotopes are a centrally symmetric class of polytopes that mitigate the curse of dimensionality and enable scalable set computations. They are given by $\Omega = \left\{c + G_z \lambda \;\middle|\; \Vert \lambda \Vert_\infty \leq 1 \right\}$, with a center $c \in \R^{n_x}$ and a generator matrix $G_z \in \R^{n_x \times N_g}$.

For large-scale dynamics, zonotopic reachability frameworks~\cite{gruber2023scalable,schafer2024} provide efficiently computable invariant under-approximations. Scaling a zonotope reduces to scaling its generator bounds: $\gamma \Omega = \{ \gamma c + G_z \lambda \mid \Vert \lambda \Vert_\infty \leq \gamma\}$ (\cref{lem:zonotope_scaling}). Substituting this into \cref{eq:set_based_gamma}, the constraint \cref{eq:set_based_cbf_filter_set_constraint} becomes
\begin{subequations}\label{eq:zonotope_reformulation}
	\begin{align}
		Ax + Bu & = \gamma^+ c +  G_z \lambda , \label{eq:zonotope_reformulation_a}              \\
		        & \Vert \lambda \Vert_\infty  \leq \gamma^+, \label{eq:zonotope_reformulation_b}
	\end{align}
\end{subequations}
with variables $\lambda \in \R^{N_g}$ and $\gamma^+ \geq 0$.


\subsection{Feasible Sets of Predictive Control Problems}\label{subsec:predictive_safe_sets}
The feasible domain of a robust predictive control problem is robust positive control invariant~\cite{rawlings2017model,wabersich2023}, which makes it a useful implicitly defined basis for a set-based \gls{CBF}. Scaling the sequence-space constraints scales the corresponding set of feasible initial conditions.

Consider the standard robust predictive control formulation evaluated over horizon $N$:
\begin{subequations}
	\label{eq:psf_nominal}
	\begin{align}
		\min_{\{u_{i}\}_{i=0}^{N-1},\{x_{i}\}_{i=0}^{N}} \quad & G(u_0, u_{\mathrm{des}}) \label{eq:psf_nominal_objective} \\
		\mathrm{s.t.}   \qquad\quad                            & x_0=x,\label{eq:psf_nominal_initial_condition} \\
		\quad                                                  & x_{N} \in \gamma \mathcal X_f,\label{eq:psf_nominal_terminal_constraint} \\
		                                                       & \text{for all }i=0,\dots,N-1: \nonumber \\
		\quad                                                  & \quad x_{i+1} = Ax_{i} + Bu_{i}, \label{eq:psf_nominal_dynamics} \\
		\quad                                                  & \quad x_{i} \in \gamma \mathcal X_i, \label{eq:psf_nominal_state_constraints} \\
		\quad                                                  & \quad u_{i} \in \gamma \mathcal U_i. \label{eq:psf_nominal_input_constraints}
	\end{align}
\end{subequations}
Setting $\gamma=1$ recovers the nominal predictive safety filter structure~\cite{wabersich2023}. Robust positive invariance is ensured by tightening the intermediate constraints $\mathcal X_i, \mathcal U_i$ and imposing a robust positively invariant terminal set $\mathcal X_f$~\cite[Sec.~15]{borrelli2017predictive}.

The set of valid initial states $x$ in~\eqref{eq:psf_nominal_initial_condition} is defined implicitly by the existence of a feasible sequence $(x_0, \{u_i\}_{i=0}^{N-1})$ and denoted by $\mathcal X_N^\gamma$. Since scaling the high-dimensional constraint set linearly scales its state-space projection (\cref{lem:scaling_of_feasible_set}), we have $\gamma\mathcal X_N^1=\mathcal X_N^\gamma$. The set-based \gls{CBF} constraint \cref{eq:set_based_cbf_filter_set_constraint} can therefore be embedded into the prediction sequence as:
\begin{subequations}
	\begin{align*}
		x_0   & = Ax+Bu,                               \\
		\quad & x_{N} \in \gamma^+ \mathcal X_f,       \\
		      & \text{for all }i=0,\dots,N-1:           \\
		\quad & \quad x_{i+1} = Ax_{i} + Bu_{i},       \\
		\quad & \quad x_{i} \in \gamma^+ \mathcal X_i, \\
		\quad & \quad u_{i} \in \gamma^+ \mathcal U_i.
	\end{align*}
\end{subequations}
This integrates a robust \gls{MPC} feasible region directly into the \gls{CBF} safety filter. Compared to standard predictive safety mechanisms~\cite{wabersich2023}, the augmented architecture adds direct intervention tuning via $\alpha$, asymptotic domain stabilization (\cref{thm:set_based_cbf}), and the option to eliminate online optimization through learning (\cref{subsec:approximation_set_based_cbf}).

\subsection{Approximate Set-based \texorpdfstring{\gls{CBF}}{CBF}}\label{subsec:approximation_set_based_cbf}
The exact reformulations above rely on tractable convex optimization. High-frequency embedded platforms may, however, preclude online solvers. In this case, the set-based \gls{CBF} can be approximated offline via imitation learning, similar to~\cite{didier2023approximate}.

The geometric properties of the set-based \gls{CBF}, in particular continuity and concavity ($h(x) = 1-\gamma_{\Omega}(x)$), together with the positive homogeneity of $\gamma_\Omega$, make it a suitable target for neural network approximation. We learn a surrogate functional $\bar\gamma_{\Omega}(x)$ via dense offline sampling and use $\bar h(x)=1-\bar\gamma_{\Omega}(x)$ directly inside the generic \gls{CBF} filter~\cref{eq:cbf_filter}. Since $\bar h$ is available in closed form, the bi-level evaluation of the successor barrier that motivated the auxiliary variable $\gamma^+$ and the set-scaling constraint~\cref{eq:set_based_cbf_filter_set_constraint} in~\cref{eq:set_based_cbf_filter} is no longer required: the online problem reduces to a low-dimensional program in $u$ alone, with a single closed-form inequality $\bar h(Ax+Bu)-\bar h(x)\geq\Delta h(x)$. For small $n_u$, this can be solved without an online convex solver by an algebraic feasibility check at $u_{\mathrm{des}}$ followed, if necessary, by a local search over $u$ (e.g., the gridded scalar-input search used in~\cref{subsec:motion_control}); this avoids any online optimization but scales poorly to higher input dimensions and is therefore best suited to single- or few-input safety filters on embedded targets.

Safety guarantees then change from hard invariance to Input-to-State Stability (ISS): an approximation error bounded by $|\gamma_{\Omega}(x)-\bar\gamma_{\Omega}(x)|\leq\epsilon$ ensures that closed-loop trajectories converge to an $\epsilon$-neighborhood of $\Omega$~\cite{didier2023approximate}. In practice, Lipschitz-bounded network architectures and post-training verification via finite gridding certify $\epsilon$.

\subsection{Representation Selection}
The reformulations above cover distinct operating regimes while preserving the guarantees of \cref{prop:nominal_set_cbf,thm:set_based_cbf}. H-polytopes (\cref{subsec:poly_h_rep}) are well-suited to low-dimensional systems with microsecond execution budgets and allow explicit offline \gls{QP} solutions (\cref{sec:electrical_machine}). Zonotopes (\cref{subsec:zonotopic_set_based_cbf}) scale to intermediate dimensions ($n_x\sim 10$-$50$) at the cost of millisecond online evaluations (\cref{subsec:inverted_pendulum}). The predictive formulation (\cref{subsec:predictive_safe_sets}) is appropriate when an explicit geometric set is unavailable or too complex. The learned approximation (\cref{subsec:approximation_set_based_cbf}) serves as a fallback when computational budgets fall below what real-time convex optimization allows, trading hard invariance for certified ISS bounds.

\section{Numerical Simulations}\label{sec:numerical_simulation}

We evaluate the set-based CBF framework on two simulation studies that highlight scalability and tunability. Rather than claiming universal superiority over existing predictive or CBF-based filters, these examples isolate the core capability outlined in Table~\ref{tab:filter_comparison}: turning a static set certificate into a tunable, dynamic CBF filter. \cref{subsec:inverted_pendulum} uses a zonotopic invariant set for a chain of inverted pendulums ($n_x=10$, $n_u=5$) and verifies asymptotic stability under robust invariance (\cref{thm:set_based_cbf}). \cref{subsec:motion_control} considers the predictive-set case and contrasts the exact set-based CBF formulation with its neural-network approximation. Both examples show how the class-$\mathcal{K}^e$ function $\alpha$ shapes the intervention behavior while preserving constraint satisfaction.

\begin{figure}[t]
	\centering
	\begin{subfigure}{\linewidth}
		\centering
		\scalebox{0.6}{
			\input{figures/example_msd_x1_x2}
			\input{figures/example_msd_x3_x4}
		}
		\caption{State trajectory and sets. Black/red marks correspond to times where $h>0$/$h<0$.}
		\label{fig:zonotope_based_cbf_sets}
	\end{subfigure}
	\begin{subfigure}{\linewidth}
		\centering
		\scalebox{0.6}{
\begin{tikzpicture}

\definecolor{crimson2143940}{RGB}{214,39,40}
\definecolor{darkgray176}{RGB}{176,176,176}
\definecolor{darkorange25512714}{RGB}{255,127,14}
\definecolor{forestgreen4416044}{RGB}{44,160,44}
\definecolor{mediumpurple148103189}{RGB}{148,103,189}
\definecolor{steelblue31119180}{RGB}{31,119,180}

\begin{axis}[
height=5cm,
width=13cm,
tick align=outside,
tick pos=left,
x grid style={darkgray176},
xlabel={time $[s]$},
xmajorgrids,
xmin=0, xmax=5,
xtick style={color=black},
y grid style={darkgray176},
ylabel={$h (x(t))$},
ymajorgrids,
ymin=-0.2, ymax=0.8,
ytick style={color=black}
]
\addplot [semithick, steelblue31119180]
table {%
0 -0.107847958260203
0.1 -0.0754935791112392
0.2 -0.0528455102772534
0.3 -0.00463014638127124
0.4 0.0720995463344322
0.5 0.0504696374075329
0.6 0.170956910932574
0.7 0.137837495700296
0.8 0.254271130512577
0.9 0.235993369655819
1 0.165195357142848
1.1 0.257579787025298
1.2 0.205238426524231
1.3 0.147421936142626
1.4 0.153883767002471
1.5 0.228631908940787
1.6 0.271268052348624
1.7 0.229777382715267
1.8 0.247701041117565
1.9 0.261519139641771
2 0.351511366944737
2.1 0.246057954198814
2.2 0.213131007897394
2.3 0.280962480038918
2.4 0.196673705418673
2.5 0.210076061934439
2.6 0.271375991225393
2.7 0.297114327238368
2.8 0.263849226326255
2.9 0.293728676048386
3 0.392422683178135
3.1 0.307601030820018
3.2 0.309160827998457
3.3 0.311457165680964
3.4 0.25428830435592
3.5 0.318544706674715
3.6 0.256097942648805
3.7 0.198613447315386
3.8 0.220624580273805
3.9 0.19707094028875
4 0.247761728398186
4.1 0.173433208612888
4.2 0.286061544520862
4.3 0.263706809443782
4.4 0.248940647801739
4.5 0.220715565655746
4.6 0.242728190079809
4.7 0.17002054904578
4.8 0.168549819190491
4.9 0.169591183519107
5 0.118713808492222
};
\addplot [semithick, darkorange25512714]
table {%
0 -0.107847958260203
0.1 -0.0754935785581892
0.2 0.0446190458025524
0.3 0.134893580181916
0.4 0.0944255023270758
0.5 0.168015368259179
0.6 0.1539416650394
0.7 0.256859009634851
0.8 0.190459743551711
0.9 0.266204192272988
1 0.335191720712358
1.1 0.461802950427441
1.2 0.605358861229743
1.3 0.570074320664439
1.4 0.552721321717069
1.5 0.469168387173065
1.6 0.499342930630989
1.7 0.537516190680802
1.8 0.451825641946886
1.9 0.496458558389794
2 0.407617220132308
2.1 0.304610177008136
2.2 0.213227118904454
2.3 0.163263921219081
2.4 0.114284701429072
2.5 0.0799992846098367
2.6 0.0559994388073732
2.7 0.0391995073697554
2.8 0.0274396458680919
2.9 0.019207700372455
3 0.0134453446596179
3.1 0.00941169500601158
3.2 0.00658818157573027
3.3 0.00461172399984044
3.4 0.00322811015166213
3.5 0.00225966426374791
3.6 0.171239038942326
3.7 0.119867323134924
3.8 0.0961108504731893
3.9 0.0742612116270025
4 0.0519828402975454
4.1 0.0454917177946204
4.2 0.0318441990249581
4.3 0.0222909234628256
4.4 0.0156036438851435
4.5 0.12709590091218
4.6 0.267523107443773
4.7 0.338869337687481
4.8 0.377831062545657
4.9 0.300648796962184
5 0.28541971088205
};
\addplot [semithick, forestgreen4416044]
table {%
0 -0.107847958260203
0.1 -0.0754935772017382
0.2 0.017348450739601
0.3 0.0121437981130387
0.4 0.0637639003829767
0.5 0.0486130368241444
0.6 0.0530070750307597
0.7 0.166544674376909
0.8 0.215468264359262
0.9 0.216558062752349
1 0.156746683369432
1.1 0.109722573855068
1.2 0.0768057971373143
1.3 0.0537640458455496
1.4 0.103505024437409
1.5 0.0724534416288954
1.6 0.050717366902589
1.7 0.0440057242090636
1.8 0.0900060086792375
1.9 0.0630042037252158
2 0.0789044599904736
2.1 0.0552331091847844
2.2 0.0386631344563061
2.3 0.0270641667097945
2.4 0.0189448736550846
2.5 0.0132613930978475
2.6 0.00928291611828314
2.7 0.00649802110972264
2.8 0.00454859007217301
2.9 0.0329587049196982
3 0.0306211530716716
3.1 0.152508071963212
3.2 0.125801109091224
3.3 0.0880607614131687
3.4 0.205969858589192
3.5 0.287106937065346
3.6 0.200974851610764
3.7 0.140682389592735
3.8 0.154232841529602
3.9 0.107962986367897
4 0.169496274563988
4.1 0.133834576657702
4.2 0.165610298577878
4.3 0.139984534787701
4.4 0.149038763878745
4.5 0.231295032397422
4.6 0.346187375097596
4.7 0.381661975413914
4.8 0.408156172021653
4.9 0.429546338524976
5 0.442759541676181
};
\addplot [semithick, crimson2143940]
table {%
0 -0.107847958260203
0.1 -0.0202985504004434
0.2 0.0618306159157963
0.3 0.156083616729415
0.4 0.228013819062367
0.5 0.255481399046319
0.6 0.216548405376727
0.7 0.29831618368427
0.8 0.427161102428983
0.9 0.389117250138485
1 0.311576400367188
1.1 0.388250913704326
1.2 0.437878714219536
1.3 0.441443039522878
1.4 0.365876450564447
1.5 0.277557337325442
1.6 0.380945652208701
1.7 0.384796345508159
1.8 0.306787174228931
1.9 0.355393852438749
2 0.287522862005546
2.1 0.387000170189636
2.2 0.358059815462182
2.3 0.46658064177243
2.4 0.328227534843873
2.5 0.266873760812886
2.6 0.349152859715417
2.7 0.25416801786806
2.8 0.25195579858342
2.9 0.32310389513455
3 0.252548718157223
3.1 0.195060514914986
3.2 0.153667230480774
3.3 0.107567058093866
3.4 0.075629187000662
3.5 0.0795309728843788
3.6 0.0556716545241164
3.7 0.0389701113889898
3.8 0.0272790755711373
3.9 0.0190953286082897
4 0.013366723720634
4.1 0.00935669419639795
4.2 0.00654968243939902
4.3 0.00458477533357093
4.4 0.0032092887217664
4.5 0.00224648139398531
4.6 0.00157240804591097
4.7 0.00110067486143917
4.8 0.000770369256834602
4.9 0.000539201140963175
5 0.174639079644614
};
\addplot [semithick, mediumpurple148103189]
table {%
0 -0.107847958260203
0.1 -0.0754935741712579
0.2 -0.00718000125638096
0.3 0.135756784021108
0.4 0.0959885783137964
0.5 0.176977699301869
0.6 0.307896908974093
0.7 0.243045827842442
0.8 0.303265329872799
0.9 0.260303781258835
1 0.185105294953683
1.1 0.134415025323746
1.2 0.127408877191642
1.3 0.203608613142129
1.4 0.316668339507459
1.5 0.338552560529106
1.6 0.436811253965662
1.7 0.412821320977197
1.8 0.344004137699886
1.9 0.372857373666978
2 0.323492662426039
2.1 0.47377732809579
2.2 0.509871808595858
2.3 0.597333500886736
2.4 0.692454718166886
2.5 0.610996158078673
2.6 0.527143680543772
2.7 0.548046744921341
2.8 0.556654453465794
2.9 0.659418211383469
3 0.700563702081993
3.1 0.728046483302
3.2 0.734525429707851
3.3 0.67255663594158
3.4 0.618285431438253
3.5 0.661296037484072
3.6 0.70333063944765
3.7 0.671468074864334
3.8 0.705882098631288
3.9 0.638276626266602
4 0.587747607020712
4.1 0.648669030936586
4.2 0.659525859385418
4.3 0.574483120886854
4.4 0.587826920177738
4.5 0.660135732487432
4.6 0.726266767392873
4.7 0.614512126036669
4.8 0.684861622215864
4.9 0.721624676057001
5 0.649181192844747
};
\end{axis}

\end{tikzpicture}
		}
		\caption{Evolution of the zonotope-based \glspl{CBF}~$h$ for different $u_{\mathrm{des}}(k)$.}
		\label{fig:zonotope_based_cbf_trajectories}
	\end{subfigure}
	\caption{Zonotope-based CBF simulation: pendulum chain ($n_x=10$). The CBF asymptotically stabilizes $\Omega$ from initial conditions outside the safe set.}
	\label{fig:zonotope_combined}
\end{figure}

\subsection{Chain of Inverted Pendulums}\label{subsec:inverted_pendulum}

To demonstrate the synthesis workflow on a higher-dimensional system, we consider a chain of $5$~inverted pendulums~\cite[Sec.~5]{conte2012}, yielding $n_x = 10$ and $n_u = 5$, with a sampling time of $T_s = 0.1\,\si{\second}$.
The system is subject to state and input constraints $\|x\|_\infty \leq 10$ and $\|u\|_\infty \leq 10$.
To invoke the recovery guarantees of \cref{thm:set_based_cbf}, we introduce the additive disturbance set $\mathcal W = [-1,1]^{10}$ and formulate the convex optimization problem using CVXPY~\cite{diamond2016cvxpy} with the Clarabel solver~\cite{clarabel2024}.

The offline design benefits from efficient zonotopic set computations: combining the methods of \cite{gruber2023scalable} and \cite{schafer2024}, a 50-generator zonotopic under-approximation of the maximal robust control invariant set~$\tilde\Omega$ is obtained in roughly \SI{1}{\second}. By contrast, computing the maximal (non-robust) control invariant set with classical polytopic methods~\cite{herceg2013multi} exceeded a \SI{24}{\hour} timeout. Online, the safety-filter \gls{QP} (via the zonotopic reformulation in \cref{subsec:zonotopic_set_based_cbf}) has $n_u + N_g = 55$ decision variables and evaluates in \SI{3.4}{\milli\second} on average ($\sigma = \SI{1.7}{\milli\second}$) per time step. The step bound is parameterized as $\alpha(r) = s \cdot r$ with $s = 0.3$.

\cref{fig:zonotope_based_cbf_sets} shows the simulation results for $\Omega = 0.9\,\tilde\Omega$, initialized at $x_0 = \begin{bmatrix} -9 & 9 & 0 & \ldots & 0\end{bmatrix}^T \notin \Omega$, with the desired input $u_{\mathrm{des}}(k) \in \R^5$ sampled uniformly from $\mathcal U$ at each step. \cref{fig:zonotope_based_cbf_trajectories} traces the evolution of the zonotope-based \glspl{CBF}~$h$ across these sequences. The trajectories consistently converge to the safe set, and once interior, the prescribed damping dynamics govern any subsequent boundary approach.

\subsection{Motion Control}\label{subsec:motion_control}

To assess practical applicability, we consider a lateral motion-control task for an autonomous vehicle. The plant is a discrete-time linearized single-track vehicle with six states $x = [\psi_e, \dot\psi_e, \beta, y_e, \delta, \dot\delta]^\top$ and a single input $u = \delta_{\mathrm{des}}$. Operating at \SI{30}{\kilo\meter\per\hour} with a \SI{50}{\milli\second} sampling time, the system must satisfy box constraints, including $|y_e|\leq 0.3\,\si{\meter}$, $|\delta|\leq 0.26\,\si{\radian}$, and $|u|\leq \pi/15\,\si{\radian}$. The set-based \gls{CBF} is defined implicitly via the feasible set of a nominal predictive controller ($N=80$, discrete-time LQR terminal set~\cite{borrelli2017predictive}) as in \cref{subsec:predictive_safe_sets}. We compare this exact formulation against a neural-network approximation (\cref{subsec:approximation_set_based_cbf}).

\cref{fig:motion_control} shows the CBF topography over the state space (left) and the closed-loop response to a constant, unsafe steering command (right). The neural-network approximation ($2\times 2048$ neurons, ReLU, $10^6$ samples~\cite{paszke2017automatic}) is evaluated for $s\in\{0.05, 0.1\}$. According to~\eqref{eq:tuning_bound}, $s=0.05$ enforces early, heavily damped intervention, while $s=0.1$ permits a more aggressive boundary approach. These relatively low values reflect the slow \SI{50}{\milli\second} sampling rate relative to the vehicle dynamics, in line with the slow-sampled tuning guidance in \cref{sec:set_based_cbf}. The exact filter (Acados SQP-RTI, $N=80$) solves in \SI{5.95}{\milli\second}/step and satisfies the constraints; the neural approximation (evaluating in \SI{5.47}{\milli\second}/step on a 128-point grid) reproduces the targeted intervention profile. The approximation is attractive when the online solver is the bottleneck: it removes the iterative QP/SQP loop in favor of a fixed-cost, branch-free forward pass with deterministic timing and no warm-start or feasibility-recovery logic, batches naturally on parallel hardware (here, ~$128\times$ faster per query than the exact CBF evaluation), and scales to embedded targets without a numerical solver in the loop, at the price of the bounded-error ISS-style safety guarantee discussed in \cref{subsec:approximation_set_based_cbf}.

\begin{figure}
	\centering
	\includegraphics[width=0.5\linewidth]{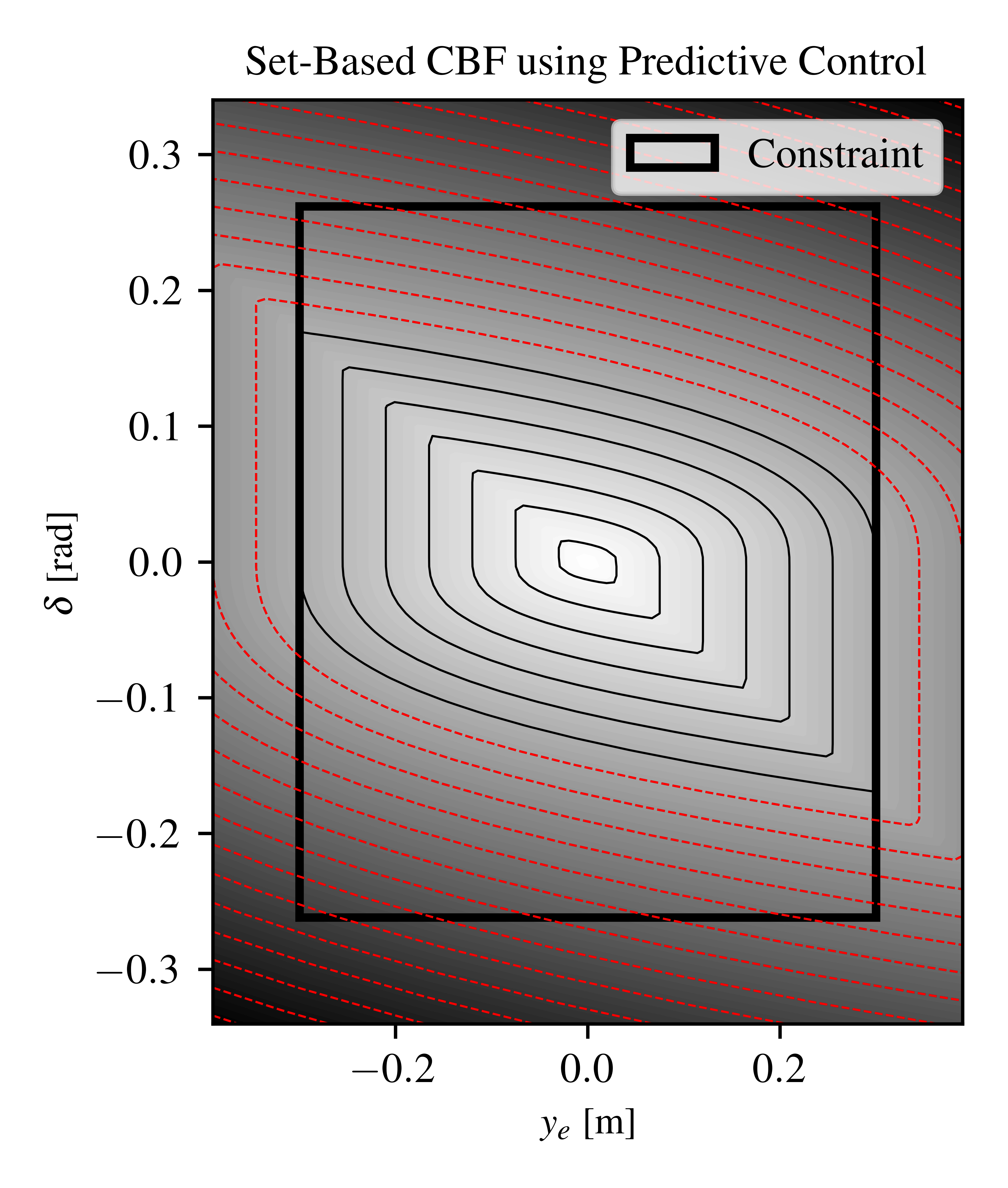}~
	\includegraphics[width=0.5\linewidth]{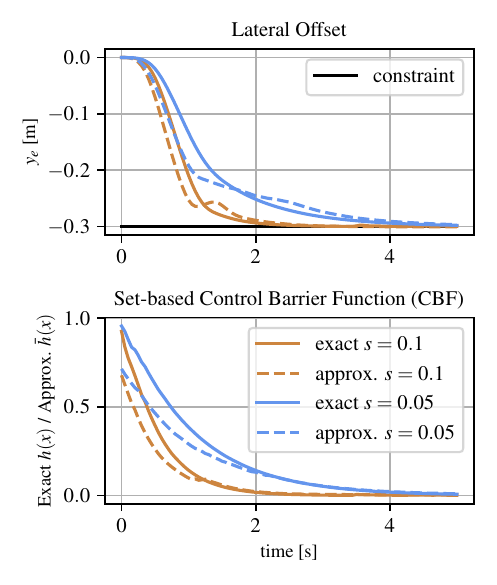}
	\caption{Lateral motion control example. \textbf{Left:} Set-based CBF topography utilizing the feasible set of an \gls{MPC}~(\cref{subsec:predictive_safe_sets}), plotted over steering angle $\delta$ and lateral offset $y_e$. Values range from 1 (white) to darker bounds. Black contours indicate the safe region ($h(x)\geq 0$); red contours denote $h(x) < 0$. \textbf{Right:} Closed-loop trajectories under the safety filter~\eqref{eq:set_based_cbf_filter} subject to a constant unsafe desired steering angle. The exact filter is compared against the neural approximation (\cref{subsec:approximation_set_based_cbf}) across different damping parameters $s$. Following~\eqref{eq:tuning_bound}, $s=0.05$ restricts the per-step margin decay to $5\%$, yielding earlier and smoother intervention than $s=0.1$.}
	\label{fig:motion_control}
\end{figure}

\section{Experiment: Electric Motor}\label{sec:electrical_machine}

\begin{figure}
	\centering
	\includegraphics*[width=\linewidth]{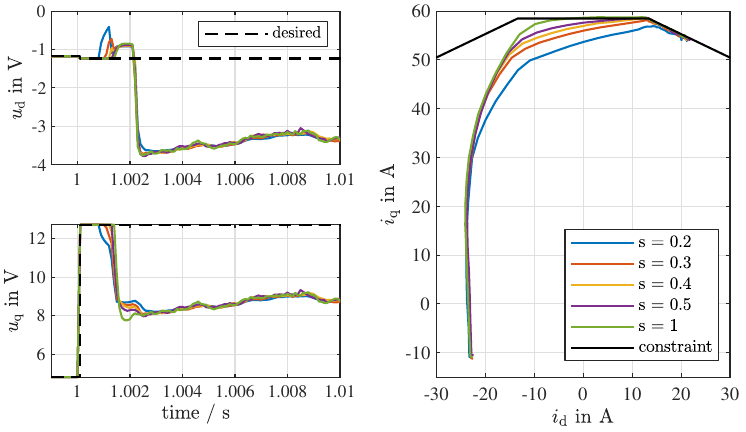}
	\caption{Experimental test bench results for an electric machine subject to a step change in desired input voltages at $t = \SI{1}{\second}$. The \gls{CBF} step bound parameter $s\in[0.2,1]$ defines $\alpha(h(x))=s\cdot h(x)$ in~\eqref{eq:cbf_step_bound}. Following~\eqref{eq:tuning_bound}, $s=0.2$ bounds the per-step safety-margin decay to $20\%$, yielding early, smooth intervention. Conversely, $s=1$ allows the margin to immediately exhaust to zero, recovering the standard invariance-filter baseline.}
	\label{fig:testbench_cbf_emachine}
\end{figure}

This section moves the set-based \gls{CBF} formulation to physical hardware and addresses the deployment challenges highlighted in \cref{tab:filter_comparison}. Specifically, we map an invariant set for a constrained linear system into a tunable \gls{CBF} filter that executes within a microsecond sampling budget.
The full workflow (\cref{sec:design_and_implementation}) is deployed on a \gls{PMSM} test bench: (i)~compute the maximal control invariant set for the dq-current model (\cref{subsec:poly_h_rep}), (ii)~formulate the set-based \gls{CBF} filter as a parametric \gls{QP}~\eqref{eq:cbf_polytope_constraint}, (iii)~solve the parametric \gls{QP} offline to guarantee execution within a \SI{150}{\micro\second} budget, and (iv)~tune the intervention parameter~$s$ directly on the physical system.

The \gls{PMSM} operates at a constant $1000\,\si{rpm}$ and is modeled as an LTI system~\eqref{eq:linear_system} with states $x = [i_{\mathrm{d}},\; i_{\mathrm{q}}]^\top$ and inputs $u = [u_{\mathrm{d}},\; u_{\mathrm{q}}, 1]^\top$, subject to polytopic constraints $\mathcal{X}$ and $\mathcal{U}$; see \cite{rosch2023predictive} for modeling details. The set-based \gls{CBF} uses the maximal positive control invariant set in H-representation and forms $h(x) = 1 - \gamma_{\Omega}(x)$. With the step bound $\alpha(h(x)) = s\cdot h(x)$ for $s \in [0.2, 1]$, the filter reduces to a parametric \gls{QP} in $\theta = [x^\top,\; u_{\mathrm{des}}^\top,\; h(x)]^\top$. Using the MPT toolbox~\cite{herceg2013multi}, this is solved explicitly offline and yields a piecewise affine control law partitioned into 127 regions.

On dSPACE1007 hardware, the mean evaluation time ranges from \SI{14.5}{\micro\second} ($s=1$) to \SI{23.7}{\micro\second} ($s=0.5$). Since the explicit law relies only on region localization and affine evaluations, the execution time is deterministic. The theoretical worst-case evaluation across all 127 regions is bounded by \SI{28.04}{\micro\second}, leaving \SI{121.96}{\micro\second} for ADC/PWM operations. As shown in \cref{fig:testbench_cbf_emachine}, decreasing $s$ produces earlier and smoother intervention, while $s=1$ recovers the strict invariance baseline. By satisfying (P1)-(P4) (\cref{sec:problem_setting}), this approach replaces the iterative online horizon optimization of predictive baselines~\cite{rosch2023predictive} with a single, deterministically bounded explicit evaluation. Structurally, the difference is that the predictive filter tunes behavior through horizon, cost, and terminal-set choices solved online, whereas the set-based \gls{CBF} fixes the invariant geometry offline and exposes the scalar parameter $s$ for commissioning-time adjustment.

The experiment thus confirms the two properties most relevant for embedded deployment: certifiable timing and transparent single-parameter tuning. The out-of-bound recovery mechanism of \cref{thm:set_based_cbf} is not exercised here, since the enforced inner current loop keeps the state within $\Omega$; its effectiveness is instead demonstrated in the pendulum-chain study (\cref{subsec:inverted_pendulum}).

\section{Conclusion}\label{sec:conclusion}
We presented a practical workflow that turns any available control invariant set for a linear constrained system into a convex, real-time safety filter with a single class-$\mathcal{K}^e$ tuning parameter. The construction inherits the offline scalability of established set-computation tools (polytopic, zonotopic, or \gls{MPC}-implicit) and provides constraint satisfaction within $\Omega$, asymptotic recovery from $\tilde\Omega\setminus\Omega$, and a learning-based approximation with bounded-error (ISS) safety. Hardware validation on a \gls{PMSM} drive achieves worst-case \SI{28.04}{\micro\second} per evaluation within a \SI{150}{\micro\second} budget and confirms that the tuning parameter $s$ shapes closed-loop behavior in a physically interpretable way.
\begin{appendices}
	\section{Proof of \cref{prop:nominal_set_cbf}}
	The proof verifies the conditions of \cref{def:barrier_function} by assembling known properties of the Minkowski functional (parts a-b) with a scaling argument for control invariant sets (part c), and then showing forward invariance directly (part~d).
	Parts a), b), and c) show basic properties of the set-based \gls{CBF}, leading to invariance in part d):
	\paragraph*{a) Safe set $\Omega$ and $h(x)$} We show that the control invariant set $\Omega$ equals the safe set of \cref{def:barrier_function}, i.e., $x\in\Omega \Leftrightarrow h(x) \geq 0$. Since $\gamma_{\Omega}(x)$ is the Minkowski functional and $\Omega$ is convex with $0\in\mathrm{int}(\Omega)$, we have $\Omega = \{x ~|~ \gamma_{\Omega}(x) \leq 1\}$~\cite[Section 3.1.2]{blanchini2015set}. Therefore, $x\in\Omega \Leftrightarrow \gamma_{\Omega}(x) \leq 1 \Leftrightarrow h(x) = 1 - \gamma_{\Omega}(x) \geq 0$. Moreover, $\gamma_{\Omega}(x)$ is well-defined for all $x\in\R^{n_x}$ since $0\in\mathrm{int}(\Omega)$.
	\paragraph*{b) Continuity and concavity of $h(x)$}\label{par:proof_implicit_cbf_continuity}
	Since $\gamma_{\Omega}(x)$ is the Minkowski functional of the compact and convex set $\Omega$ with $0\in\mathrm{int}(\Omega)$, it is continuous and convex on $\R^{n_x}$~\cite[Section 3.1.2]{blanchini2015set}. Therefore, $h(x) = 1 - \gamma_{\Omega}(x)$ is continuous and concave on $\R^{n_x}$.
	\paragraph*{c) For all $x\in \Omega$ the bound \cref{eq:cbf_step_bound} holds}
	For all $x\in\Omega$, we have $x\in\gamma_{\Omega}(x)\Omega$. The scaled set $\gamma_{\Omega}(x)\Omega$ is control invariant because for any $y\in\gamma_{\Omega}(x)\Omega$ one can write $y=\gamma_{\Omega}(x)z$ with $z\in\Omega$; control invariance of $\Omega$ provides $u\in\mathcal U$ with $Az+Bu\in\Omega$, and linearity gives $Ay+B(\gamma_{\Omega}(x)u)=\gamma_{\Omega}(x)(Az+Bu)\in\gamma_{\Omega}(x)\Omega$ with $\gamma_{\Omega}(x)u\in\mathcal U$ since $\gamma_{\Omega}(x)\leq 1$ and $0\in\mathcal U$~\cite[Section~3.1.2]{blanchini2015set}. Therefore, there exists an input $u\in\mathcal U$ such that $Ax+Bu\in\gamma_{\Omega}(x)\Omega$, implying $\gamma_{\Omega}(Ax+Bu)\leq \gamma_{\Omega}(x)$ and therefore $h(Ax+Bu)\geq h(x)$. Since $\Delta h(x)\leq 0$ for all $x\in\Omega$, the condition $h(Ax+Bu) - h(x) \geq \Delta h(x)$ in~\cref{eq:cbf_condition} is satisfied, explicitly satisfying the first case of~\cref{eq:cbf_step_bound}.
	\paragraph*{d) Forward invariance of $\Omega$} From part a), $\Omega=\{x:h(x)\geq0\}$. For any $x\in\Omega$ and any input $u\in K_{\mathrm{CBF}}(x)$, \cref{eq:cbf_condition} gives $h(Ax+Bu)\geq h(x)+\Delta h(x)$. Since $h(x)\geq0$, the first case of~\cref{eq:cbf_step_bound} implies $h(x)+\Delta h(x)=h(x)-\min(\alpha(h(x)),h(x))\geq0$. Hence $Ax+Bu\in\Omega$. Applying this implication recursively shows that $\Omega$ is forward invariant for any policy selecting inputs from $K_{\mathrm{CBF}}(x)$ and any initial condition $x(0)\in\Omega$.
	\section{Technicalities}
	\begin{lemma}\label{lem:zonotope_scaling}
		For a zonotope $\Omega = \{x = c + G_z \lambda ~|~ \Vert \lambda \Vert_\infty \leq 1\}$ and $\gamma \geq 0$, the scaled zonotope is given by $\gamma \Omega = \{x = \gamma c + G_z \lambda ~|~ \Vert \lambda \Vert_\infty \leq \gamma\}$.
	\end{lemma}
	\begin{proof}
		By definition of the zonotopic representation, we have
		\begin{align*}
			\gamma \Omega
			 & = \left\{x = \gamma (c + G_z \lambda) ~\middle|~ \Vert \lambda \Vert_\infty \leq 1 \right\}                                                     \\
			 & = \left\{x = \gamma c + G_z (\gamma \lambda) ~\middle|~ \Vert \lambda \Vert_\infty \leq 1 \right\}                                              \\
			 & = \left\{x = \gamma c + G_z \tilde{\lambda} ~\middle|~ \tilde{\lambda} = \gamma \lambda, \Vert \lambda \Vert_\infty \leq 1 \right\}             \\
			 & = \left\{x = \gamma c + G_z \tilde{\lambda} ~\middle|~ \tilde{\lambda} = \gamma \lambda, \Vert \gamma \lambda \Vert_\infty \leq \gamma \right\} \\
			 & = \left\{x = \gamma c + G_z \tilde{\lambda} ~\middle|~ \Vert \tilde{\lambda} \Vert_\infty \leq \gamma \right\},
		\end{align*}
		where we assume $\gamma \geq 0$. This reformulation avoids products of decision variables ($\gamma$ and $\lambda$), yielding the convex constraint~\cref{eq:zonotope_reformulation_b}.
	\end{proof}

	\begin{lemma}\label{lem:scaling_of_feasible_set}
		For all $\gamma\geq 0$ it holds that $\gamma\mathcal X_N^1=\mathcal X_N^\gamma$.
	\end{lemma}
	\begin{proof}
		In the case that $\gamma=0$, it follows from \cref{eq:psf_nominal_initial_condition} and \cref{eq:psf_nominal_state_constraints} that $\mathcal X^0_N=\{0\}$, which equals $0\mathcal X^1_N$ by definition.
		For $\gamma>0$, we have
		$\gamma \mathcal Z_N^1  = \{ (\gamma z, \{\gamma v_{i}\}) ~|~ (z, \{v_{i}\}) \in\mathcal Z_N^1\}$,
		where we substitute $(\gamma z, \{\gamma v_{i}\})\triangleq (x, \{u_{i}\})$ to obtain
		\begin{align}
			\gamma \mathcal Z_N^1 & = \{ (x, \{u_{i}\}) ~|~ (\gamma^{-1}x, \gamma^{-1}\{u_{i}\}) \in\mathcal Z_N^1\}.
			\label{eq:lem_scaling_inverse_relation}
		\end{align}
		Next, we rewrite~\cref{eq:lem_scaling_inverse_relation} as
		\begin{align*}
			 & \left\{ (x, \{u_{i}\}) \middle|
			\begin{matrix*}[l]
				& (A^N \gamma^{-1}x + \Sigma_{l=0}^{N-1}A^{N-l-1}B \gamma^{-1}u_{l}) \in \mathcal X_f, \\
				& \forall i\in\{0,..N-1\}:                                                               \\
				& \quad (A^i \gamma^{-1}x + \Sigma_{l=0}^{i-1}A^{i-l-1}B \gamma^{-1}u_{l}) \in \mathcal X_i,\\
				& \quad  \gamma^{-1}u_{i} \in \mathcal U_i
			\end{matrix*}
			\right\}                                   \\
			 & =\mathcal Z^\gamma_N
		\end{align*}
		using $x_{0} = x$ and by inserting the dynamics~\cref{eq:psf_nominal_dynamics}.
		Through the relation $\mathcal Z^\gamma_N= \gamma\mathcal Z^1_N$ we finally obtain
		\begin{align*}
			\mathcal X^\gamma_N & = \{x ~|~ (x, \{u_{i}\}) \in\gamma \mathcal Z_N^1\}                 \\
			                    & = \{x ~|~ (\gamma^{-1}x, \{\gamma^{-1}u_{i}\}) \in \mathcal Z_N^1\} \\
			                    & = \{\gamma z ~|~ (z, \{v_{i}\}) \in \mathcal Z_N^1\}=\gamma \mathcal X^1_N,
		\end{align*}
		showing the desired result.
	\end{proof}

\end{appendices}

\bibliographystyle{IEEEtran}
\bibliography{bib}

@article{clarabel2024,
  author    = {P. J. Goulart and Y. Chen},
  title     = {{Clarabel}: An Interior-Point Solver for Conic Programs with Quadratic Objectives},
  journal   = {Math. Program. Comput.},
  year      = {2026},
  doi       = {10.1007/s12532-026-00320-7},
  note      = {Published online}
}

@book{blanchini2015set,
  title={Set-Theoretic Methods in Control},
  author={Blanchini, Franco and Miani, Stefano},
  year={2015},
  publisher={Springer}
}

@article{liniger2017real,
  author    = {A. Liniger and J. Lygeros},
  title     = {Real-Time Control for Autonomous Racing Based on Viability Theory},
  journal   = {IEEE Trans. Control Syst. Technol.},
  volume    = {27},
  number    = {2},
  pages     = {464--478},
  year      = {2019}
}

@inproceedings{gurriet2018towards,
  author    = {T. Gurriet and A. Singletary and J. Reher and L. Ciarletta and E. Feron and A. Ames},
  title     = {Towards a Framework for Realizable Safety Critical Control Through Active Set Invariance},
  booktitle = {Proc. ACM/IEEE 9th Int. Conf. Cyber-Physical Systems (ICCPS)},
  pages     = {98--106},
  year      = {2018}
}

@inproceedings{leeman2023predictive,
  author    = {A. Leeman and J. K{\"o}hler and S. Bennani and M. Zeilinger},
  title     = {Predictive Safety Filter Using System Level Synthesis},
  booktitle = {Proc. Learning for Dynamics and Control Conf.},
  pages     = {1180--1192},
  year      = {2023}
}

@article{gruber2023scalable,
  author    = {F. Gruber and M. Althoff},
  title     = {Scalable Robust Safety Filter With Unknown Disturbance Set},
  journal   = {IEEE Trans. Autom. Control},
  volume    = {68},
  number    = {12},
  pages     = {7756--7770},
  year      = {2023}
}

@article{diamond2016cvxpy,
  author    = {S. Diamond and S. Boyd},
  title     = {{CVXPY}: {A} {P}ython-Embedded Modeling Language for Convex Optimization},
  journal   = {J. Mach. Learn. Res.},
  volume    = {17},
  number    = {83},
  pages     = {1--5},
  year      = {2016}
}

@inproceedings{choi2021robust,
  author    = {J. J. Choi and D. Lee and K. Sreenath and C. J. Tomlin and S. L. Herbert},
  title     = {Robust Control Barrier-Value Functions for Safety-Critical Control},
  booktitle = {Proc. 60th IEEE Conf. Decision and Control (CDC)},
  pages     = {6814--6821},
  year      = {2021}
}

@article{ames2016control,
  author    = {A. D. Ames and X. Xu and J. W. Grizzle and P. Tabuada},
  title     = {Control Barrier Function Based Quadratic Programs for Safety Critical Systems},
  journal   = {IEEE Trans. Autom. Control},
  volume    = {62},
  number    = {8},
  pages     = {3861--3876},
  year      = {2017}
}

@article{fisac2018general,
  author    = {J. F. Fisac and A. K. Akametalu and M. N. Zeilinger and S. Kaynama and J. Gillula and C. J. Tomlin},
  title     = {A General Safety Framework for Learning-Based Control in Uncertain Robotic Systems},
  journal   = {IEEE Trans. Autom. Control},
  volume    = {64},
  number    = {7},
  pages     = {2737--2752},
  year      = {2019}
}

@inproceedings{wabersich2018linear,
  author    = {K. P. Wabersich and M. N. Zeilinger},
  title     = {Linear Model Predictive Safety Certification for Learning-Based Control},
  booktitle = {Proc. IEEE Conf. Decision and Control (CDC)},
  pages     = {7130--7135},
  year      = {2018}
}

@article{rosolia2021robust,
  author    = {U. Rosolia and X. Zhang and F. Borrelli},
  title     = {Robust Learning Model-Predictive Control for Linear Systems Performing Iterative Tasks},
  journal   = {IEEE Trans. Autom. Control},
  volume    = {67},
  number    = {2},
  pages     = {856--869},
  year      = {2022}
}

@article{schafer2024,
  author    = {L. Schäfer and F. Gruber and M. Althoff},
  title     = {Scalable Computation of Robust Control Invariant Sets of Nonlinear Systems},
  journal   = {IEEE Trans. Autom. Control},
  volume    = {69},
  number    = {2},
  pages     = {755--770},
  year      = {2024},
  doi       = {10.1109/TAC.2023.3275305}
}

@inproceedings{conte2012,
  author    = {C. Conte and N. R. Voellmy and M. N. Zeilinger and M. Morari and C. N. Jones},
  title     = {Distributed Synthesis and Control of Constrained Linear Systems},
  booktitle = {Proc. Amer. Control Conf. (ACC)},
  pages     = {6017--6022},
  year      = {2012},
  doi       = {10.1109/ACC.2012.6314654}
}

@inproceedings{herceg2013multi,
  author    = {M. Herceg and M. Kvasnica and C. N. Jones and M. Morari},
  title     = {Multi-Parametric Toolbox 3.0},
  booktitle = {Proc. Eur. Control Conf. (ECC)},
  pages     = {502--510},
  year      = {2013}
}

@article{wabersich2021soft,
  author    = {K. P. Wabersich and R. Krishnadas and M. N. Zeilinger},
  title     = {A Soft Constrained {MPC} Formulation Enabling Learning From Trajectories With Constraint Violations},
  journal   = {IEEE Control Syst. Lett.},
  volume    = {6},
  pages     = {980--985},
  year      = {2022}
}

@article{alan2023control,
  author    = {A. Alan and A. J. Taylor and C. R. He and A. D. Ames and G. Orosz},
  title     = {Control Barrier Functions and Input-to-State Safety With Application to Automated Vehicles},
  journal   = {IEEE Trans. Control Syst. Technol.},
  volume    = {31},
  number    = {6},
  pages     = {2744--2759},
  year      = {2023}
}

@book{rawlings2017model,
  author    = {J. B. Rawlings and D. Q. Mayne and M. Diehl},
  title     = {Model Predictive Control: Theory, Computation, and Design},
  edition   = {2},
  year      = {2017},
  publisher = {Nob Hill Publishing}
}

@inproceedings{paszke2017automatic,
  author    = {A. Paszke and S. Gross and S. Chintala and G. Chanan and E. Yang and Z. DeVito and Z. Lin and A. Desmaison and L. Antiga and A. Lerer},
  title     = {Automatic Differentiation in {PyTorch}},
  booktitle = {Proc. NIPS-W},
  year      = {2017}
}

@inproceedings{didier2023approximate,
  author    = {A. Didier and R. C. Jacobs and J. Sieber and K. P. Wabersich and M. N. Zeilinger},
  title     = {Approximate Predictive Control Barrier Functions Using Neural Networks: A Computationally Cheap and Permissive Safety Filter},
  booktitle = {Proc. Eur. Control Conf. (ECC)},
  pages     = {1--7},
  year      = {2023}
}

@inproceedings{agrawal2017discrete,
  author    = {A. Agrawal and K. Sreenath},
  title     = {Discrete Control Barrier Functions for Safety-Critical Control of Discrete Systems With Application to Bipedal Robot Navigation},
  booktitle = {Proc. Robotics: Science and Systems},
  volume    = {13},
  pages     = {1--10},
  year      = {2017}
}

@book{borrelli2017predictive,
  author    = {F. Borrelli and A. Bemporad and M. Morari},
  title     = {Predictive Control for Linear and Hybrid Systems},
  year      = {2017},
  publisher = {Cambridge University Press}
}

@article{wabersich2021probabilistic,
  author    = {K. P. Wabersich and L. Hewing and A. Carron and M. N. Zeilinger},
  title     = {Probabilistic Model Predictive Safety Certification for Learning-Based Control},
  journal   = {IEEE Trans. Autom. Control},
  volume    = {67},
  number    = {1},
  pages     = {176--188},
  year      = {2022}
}

@article{wabersich22,
  author    = {K. P. Wabersich and M. N. Zeilinger},
  title     = {Predictive Control Barrier Functions: Enhanced Safety Mechanisms for Learning-Based Control},
  journal   = {IEEE Trans. Autom. Control},
  volume    = {68},
  number    = {5},
  pages     = {2638--2651},
  year      = {2023}
}

@article{wabersich2023,
  author    = {K. P. Wabersich and A. J. Taylor and J. J. Choi and K. Sreenath and C. J. Tomlin and A. D. Ames and M. N. Zeilinger},
  title     = {Data-Driven Safety Filters: {H}amilton-{J}acobi Reachability, Control Barrier Functions, and Predictive Methods for Uncertain Systems},
  journal   = {IEEE Control Syst. Mag.},
  volume    = {43},
  number    = {5},
  pages     = {137--177},
  year      = {2023}
}

@inproceedings{rosch2023predictive,
  author    = {D. R{\"o}sch and F. Berkel and M. L{\"o}hning and M. Manderla and R. Soloperto and F. Allg{\"o}wer},
  title     = {A Predictive Safety Filter for Safe Learning of Optimal Operation of Permanent Magnet Synchronous Motors},
  booktitle = {Proc. Eur. Control Conf. (ECC)},
  pages     = {1--6},
  year      = {2023}
}

\end{document}